\documentclass[oneside,english]{amsart}
\usepackage[T1]{fontenc}
\usepackage[latin9]{inputenc}
\usepackage{fancyhdr}
\usepackage{mathrsfs}
\usepackage{amsthm}
\usepackage{amssymb}
\PassOptionsToPackage{normalem}{ulem}
\usepackage{ulem}

\makeatletter
\numberwithin{equation}{section}
\numberwithin{figure}{section}
\theoremstyle{plain}
\newtheorem{thm}{\protect\theoremname}
\theoremstyle{plain}
\newtheorem{prop}[thm]{\protect\propositionname}

\makeatother

\usepackage{babel}
\providecommand{\propositionname}{Proposition}
\providecommand{\theoremname}{Theorem}

\begin{document}
\title{A non-equilibrium theoretical framework for statistical physics with
application to turbulent systems and their predictability}
\author{Richard Kleeman}
\date{June 2019}
\address{Courant Institute of the Mathematical Sciences. New York University.
251 Mercer Street New York.}
\email{kleeman@cims.nyu.edu}
\begin{abstract}
A new theoretical approach to non-equilibrium statistical systems
has recently been proposed by the author, a co-author and others.
It is based on a variational principle which is associated with the
discrepancy of a path through thermodynamical space to one following
Liouvillean evolution. In this contribution the approach is extended
in such a way that it can be applied to a wide range of practical
non-equilibrium statistical systems such as those arising in turbulence
but also to a general class of statistical physics models. The new
methodology allows for application to autonomous dynamical systems
generalizing the previous work which applied only to Hamiltonian systems.
Furthermore it provides a general analysis of near equilibrium conditions
which allows for a natural analysis of predictability limits in turbulent
systems. Finally it describes a method is described for the numerical
calculation of far from equilibrium thermodynamical trajectories.
\end{abstract}

\maketitle

\lhead{}

\chead{A non-equilibrium statistical physics framework}

\rhead{}

\section{Introduction}

A problem of immense practical utility concerns the evolution of slow
variables within a system with many degrees of freedom. In general
this is influenced by the remaining fast variables of the system as
well as other slow variables. Since the subset of slow variables is
often useful and small while the fast variable subset is large and
usually irrelevant, one typically requires a statistical dynamical
system dependent only on the first small number of variables. If the
fast variables are ignored then the slow variables may conveniently
be regarded as continuous random variables.

An early approach to this problem in fluid dynamical systems was that
of moment closure wherein the low order moment evolution of the slow
variables was derived using a physically based closure hypothesis.
Such a truncating closure is required because typically all moments
influence all others. A typical example of this is provided by the
extensive work of Fredriksen and co-authors (see, for example, \cite{frederiksen2004regularized}).

Another widely used approach in the fluid context idealizes the fast
variables as stochastic forcing of the slow random variables. This
can be done in a mathematically rigorous fashion and can lead to rather
complex and comprehensive stochastic differential equations. An excellent
exposition on this approach can be found in the work of Majda and
co-authors (see, for example, \cite{mtv2}).

The key initial assumption of any non-equilibrium statistical theory
concerns the identification of the relevant slow variables for the
problem. This choice will obviously be influenced by the time scale
considered important. This issue can be seen acutely in the area of
irreversible thermodynamics. The classical version of this theory
due to Onsager and many others (see, for example, \cite{de1984irreversible})
assumes that \textbf{local} energy and particle number are the slow
variables and generalizes Gibbsian equilibrium thermodynamics to consider
where only a local equilibrium holds. In recent decades this has been
found to be an inadequate description when relatively fast phenomenon
require consideration. Then one needs to also consider as slow variables,
for example, the local heat and density \textbf{fluxes} which on a
slower time scale adjust to their classical Onsagerian values as determined
by the previously mentioned slow variables. This field of study is
referred to as extended irreversible thermodynamics (EIT) and an excellent
survey can be found in \cite{jou2010extended}\footnote{It is worth observing that older perturbative approaches to kinetic
theory such as those of Chapman, Enskog and Grad (see \cite{liboff2003kinetic})
also imply that heat and momentum fluxes need independent consideration
under certain conditions. This is a complex area and is reviewed comprehensively
in the cited reference in the main text.}. A different perspective on this can be found in the GENERIC approach
documented in \cite{oett} which we remark on later. Another further
interesting recent approach is ``stochastic thermodynamics'' (see,
for example, \cite{seifert2012stochastic}) where the thermodynamical
variables are assumed to have an associated density evolving according
to a Fokker Planck equation. There is an interesting connection between
this latter approach and that to be promoted here which we discuss
more at the end of section 2.

Traditionally when one moves to consider the underlying densities
for slow variables\footnote{Often for the purpose of developing a kinetic theory}
it is common to invoke a maximum entropy principle using the expected
values of the slow variables as constraints. This is usually applied
over a ``historical'' time period leading up to the present time
of interest (see, for example, \cite{zubarev1996statistical}, \cite{zubarev1997statistical}
and \cite{luzzi2013predictive}). This then leads, once some further
simplifying assumptions are applied, to density evolution equations
which are variants of the well known Mori-Zwanzig equation (see \cite{zwanzig2001nonequilibrium}).
Due to the finite historical time interval used and the simplifying
assumptions, these equations typically have memory effects which is
unlike the fluid stochastic modeling case mentioned earlier. Comparison
of these maximum entropy densities with those derived from a numerical
simulation reveals that they are often close to\footnote{Assuming that the correct identification of slow variables has been
made.} but not exactly of maximum entropy form unlike the case of the equilibrium
Gibbs densities where classical Jaynes maxent holds very accurately.

Now if one assumes, following Zubarev, that the maximum entropy densities
of certain slow variables are good approximations for the exact non-equilibrium
densities then the former ``coarse grained'' densities will no longer
satisfy the Liouville equation which the latter must. The discrepancy
from Liouvillean evolution can be calculated with the tools of information
theory and used to define a path variational principle rather like
the action principle of classical and quantum mechanics.

This approach was proposed originally by Turkington and developed
further by the present author and others (see, for example, \cite{turkington2012optimization},
\cite{kleeman2012nonequilibrium} and \cite{Kle14}). It has shown
very promising results for the relaxation of the first two moments
of slow variables in a variety of simple inviscid turbulent systems.
The information theoretic approach can also be applied to the study
of predictability of slow variables. These have a probability density
function which relaxes toward an equilibrium density which applies
to slow variables unconstrained by initial condition data. The degree
of statistical disequilibrium of the system can be taken as a measure
of the predictability of the slow variables and can be quantified
using the relative entropy of the evolving and equilibrium densities.
A review of this predictability approach can be found in \cite{Kle11}.
This perspective on predictability fits well with the Turkington approach
since the relative entropy functional is easily calculated there.

In this publication we extend the just mentioned theoretical framework
from the inviscid case to the realistic forced dissipative case to
allow for application to more realistic fluid systems and their practical
predictability problems\footnote{We have in mind in particular systems relevant to atmosphere/ocean
science.} Furthermore we will develop a set of tools for analyzing thermodynamical
trajectories both in the near and far from equilibrium case which
will be of use in the analysis of a wide class of non-equilibrium
thermodynamical problems. As we shall see there are strong connections
between the results described and the previously mentioned more empirical
results of \cite{oett} and \cite{seifert2012stochastic}. These deserve
careful future study in practical statistical systems of various types.

The method to be used has one essential and practically important
limitation. It assumes that the slow variables density belongs to
a particular restricted class. Such a class needs to be justified
a posteriori using numerical simulations. In the specific case to
be considered this will consist of general multivariate Gaussian densities.
Considerable experience with practical turbulence models by the author
(see in particular \cite{klee04d} and \cite{kleeman05a}) justifies
this choice for the current problem however different applications
will require renewed justification and possible modification of the
class. Without an assumption of this type theoretical progress is
not possible. In addition the class used has a strong impact on the
degree of difficulty of the theoretical calculations as we shall see
below in section 4. For more discussion on this point also see subsection
2.1 below.

The format of this publication is as follows: In section 2 the relevant
material from previous work will be summarized. In section 3 the issue
of asymptotic equilibration and the consequent fundamental limits
to predictability will be analyzed. The generalization of previous
results to a forced dissipative system will then be displayed. Finally
a numerical method for obtaining far from equilibrium thermodynamical
behavior will be derived. In section 4 the family of slow variable
densities will be restricted to both a general Gaussian and one with
a diagonal covariance matrix. Explicit calculations of the theoretical
tensors required in sections 2 and 3 will then be performed with this
restriction. The result will be a complete set of analytical results
and numerical techniques for the case of Gaussian slow variable densities.
The machinery thus developed will hopefully have wide application
to many statistical physics applications and in particular to realistic
models of turbulence. Section 5 contains a summary and some final
comments.

\section{Review of relevant past results}

\subsection{Trial densities}

The approach taken by Turkington and co-workers is to assume that
the slow variable marginal densities $\hat{\varrho_{s}}$ are instantaneously
of the Zubarev maximum entropy form
\begin{equation}
\hat{\varrho}_{s}(x)=C\exp\left[\lambda_{i}A^{i}(x)\right]\label{MAXENT}
\end{equation}

while the total densities for the entire statistical system are of
the form
\[
\hat{\varrho}\left(x\right)=Z^{-1}\left(\lambda,\beta\right)\hat{\varrho}_{s}(x)\hat{\varrho}_{eq}(x,\beta)
\]

where $A_{i}$ are the slow variables; $\lambda_{i}$ and $\beta$
are generalized inverse temperatures; $Z$ is the partition function
which normalizes the density and $\hat{\varrho}_{eq}$ is the density
for the fully system which applies asymptotically in time (see below).
The maximum entropy principle constrains the expectations of $A_{i}$
to satisfy
\[
\left\langle A_{i}\right\rangle =a_{i}
\]

and the moments $a_{i}$ are Legendre transforms of the $\lambda_{i}$
at a given $\beta$. Densities of type $\hat{\varrho}$ are referred
to as trial densities. As previously mentioned they are never exact
but are expected to be very good approximations of actual evolving
marginal densities. Note that they also pin the ``thermodynamical''
variables for the problem i.e. the $a_{i}$ or the $\lambda_{i}$.
The complete set of trial densities describe a manifold coordinatized
by the thermodynamical variables. This is the domain of information
geometry (see \cite{ama00}). The asymptotic density $\hat{\varrho}_{eq}$
can be deduced in the manner of Gibbs for a Hamiltonian system but
for more general autonomous dynamical systems will usually only be
able to be determined approximately from numerical or physical experiments.
The implications of this will be discussed further in section 3 below.
Note also from our method of definition that the trial density will
become a (possibly approximate) equilibrium density when $\lambda=0$.
In section 3 we shall see that this is the only possible thermodynamical
equilibrium for the formalism to be outlined in section 2.

\subsection{Path dependent information loss}

The full density for a Hamiltonian\footnote{The generalization to an autonomous dynamical system is considered
below in section 3.} dynamical system satisfies the Liouville equation:
\[
\frac{\partial\varrho}{\partial t}+L\varrho=0
\]

with the differential operator $L$ determined by the underlying dynamical
system via
\begin{align*}
L & =C_{j}\frac{\partial}{\partial x_{j}}\\
\frac{\partial x_{k}}{\partial t} & =C_{k}(x)
\end{align*}

On the other hand the trial density just specified will in general
not satisfy this Liouville equation in the sense that any trial density
evolved according to the Liouville equation will no longer remain
within the manifold of trial densities. Indeed it is possible \cite{Kle14}
to use the relative entropy between this evolved trial density and
a general trial density to measure this discrepancy. It has the form
\begin{align}
IL & =\left(\Delta t\right)^{2}\mathscr{L}(\dot{\lambda},\lambda)+O\left(\left(\Delta t\right)^{3}\right)\label{IL}\\
\mathscr{L}(\dot{\lambda},\lambda) & \equiv\frac{1}{2}\left\langle R^{2}\right\rangle \geq0\\
R & \equiv\left(\partial_{t}+L\right)\log\hat{\varrho}=\dot{\lambda}_{i}\left(A_{i}-a_{i}\right)+\lambda_{i}LA_{i}
\end{align}

where the Liouville equation evolution time is $\Delta t$; the overdot
indicates a time derivative and where the summation convention is
employed. The function $\mathscr{L}$ has the form
\begin{align}
\mathscr{L}(\dot{\lambda},\lambda) & =\frac{1}{2}\left(\dot{\lambda}_{i}h_{ij}(\lambda)\dot{\lambda}_{j}-2\dot{\lambda}_{i}M_{i}(\lambda)+\phi(\lambda)\right)\nonumber \\
h_{ij} & \equiv\left\langle \left(A_{i}-a_{i}\right)\left(A_{j}-a_{j}\right)\right\rangle \nonumber \\
M_{i} & \equiv-\lambda_{j}\left\langle (A_{i}-a_{i})(LA_{j})\right\rangle =\left\langle LA_{i}\right\rangle \nonumber \\
\phi & \equiv\lambda_{i}\lambda_{j}\left\langle LA_{i}LA_{j}\right\rangle \label{explicitlagr}
\end{align}

A couple of important observations may be made about $\mathscr{L}$:
Firstly like a classical dynamical Lagrangian it is quadratic in $\dot{\lambda}$.
Secondly since the $\lambda$ are a thermodynamical coordinatization
of the trial density manifold, one can consider a coordinate transformation
and under this it is easily seen that $h$, $M$ and $\phi$ will
transform as a tensors, vectors and scalars respectively. Such transformations
can be very convenient for physical interpretation as well as mathematical
manipulation (see sections 3 and 4 below).

The Liouville discrepancy in (\ref{IL}) can also be regarded as the
information loss incurred in insisting that time dependent densities
remain within the trial manifold. The convergence properties of the
expansion (\ref{IL}) have not been rigorously investigated however
a heuristic argument can be made as follows: If a path is considered
in which the thermodynamical variable $\lambda$ relaxes at the rate
observed in numerical simulations (i.e. realistically) then $\mathscr{L}$
may be non-dimensionalized with this time scale $t_{r}$. The first
term of $IL$ then becomes 
\[
\left(\frac{\Delta t}{t_{r}}\right)^{2}\mathscr{L}_{non}
\]

where, because of the quadratic form of $\mathscr{L}$ in $\dot{\lambda}$
and dimensional arguments for other terms within, $\mathscr{L}_{non}\sim O(1)$.
Similar arguments apply to the higher $n$'th terms for $IL$ which
can be shown be polynomials in $\dot{\lambda}$ of order $n+1$. Thus
convergence should occur providing that $\Delta t$ is a significantly
small fraction of the typical thermodynamical relaxation time scale
$t_{r}$ perhaps having a time scale roughly that of the fastest thermodynamical
variable considered. The non-realistic paths with faster relaxation
of $\lambda$ are expected to have much higher values of $IL$ and
not be relevant to calculations of interest (see below). Notice however
that the limit $\Delta t\rightarrow0$ is not useful. $\Delta t$
must be finite.

Consider an arbitrary path $\widehat{\lambda}$ within the trial manifold
and with total information loss $TIL$. Approximating the Riemann
sum as an integral\footnote{Such an approximation will be valid providing $\Delta t<<t-t_{0}$}
we obtain
\begin{align}
TIL\left[\widehat{\lambda}\right] & =\Delta t\mathcal{S}\left[\widehat{\lambda}\right]+O(\left(\Delta t\right)^{2})\label{action}\\
\mathcal{S}\left[\widehat{\lambda}\right] & \equiv\int_{t_{0}}^{t}\mathscr{L}(\dot{\lambda},\lambda)dt\nonumber 
\end{align}

where the integral is implicitly a line integral along $\hat{\lambda}$
and $\mathcal{S}$ is the action along the path corresponding to the
Lagrangian $\mathscr{L}$. Note again also the ``peculiar'' linear
dependence on the discrepancy time step $\Delta t$.

If we disregard the higher order terms in (\ref{action}) then we
have an information loss associated with an arbitrary thermodynamic
path. It thus seems natural to consider, in analogy with classical
mechanics, the extremal path between two points in thermodynamical
space which minimizes this loss. This will be the path through the
trial manifold which best respects Liouvillian evolution. In general
though we cannot prescribe a final trial density. Rather a typical
practical situation is when we have an arbitrary initial trial density
(corresponding to a general point in thermodynamic space) and we wish
to determine the relaxation path toward an asymptotic equilibrium
density.

One approach to this problem would be to consider the unique extremal
paths between the initial condition and all possible endpoints then
use the information loss of these extremal paths to find the final
point with least information loss on it's extremal. This amounts to
a double optimization procedure. It will certainly define a unique
thermodynamic trajectory $\theta(t)$. It has however a rather peculiar
property.

Consider two future times $t_{2}>t_{1}>t_{0}$ and consider the extremal
path $\eta(t)$ between the point $\theta(t_{2})$ and the initial
prescribed point $\theta(t_{0})=\eta(t_{0})$. It turns out that $\theta(t_{1})\neq\eta(t_{1})$
and this discrepancy can often be of the same order as $\theta(t_{2})-\theta(t_{0})$.

Despite the above conceptual conundrum numerical validation studies
show that $\theta(t)$ is often a good approximation to the observed
thermodynamical evolution.

Another approach to the problem is to follow the obvious analogy between
the path information loss and an action and therefore a path integral.
This can be done by invoking a generalized Boltzmann principle to
define a (Wiener) path weight measure $W$ via
\[
W\left[\widehat{\lambda}\right]=\exp\left[-\Delta t\mathcal{S}\left[\widehat{\lambda}\right]\right]
\]

Then each possible thermodynamical endpoint $\lambda(t_{2})$ can
be assigned a weighting according to the (path) integral of $W$ over
all possible paths from the prescribed initial point $\theta(t_{0})$
to the considered endpoint. This differs from the first optimization
method where only the extremal path $\eta(t)$ is considered. Obviously
though the methods coincide in the limit of large $\Delta t$.

This is entirely analogous to the Feynman path integral formalism\footnote{Strictly it is the Wick rotated Wiener path integral which is actually
better defined mathematically}. Thermodynamical points are then assigned a non-negative ``consistency''
distribution which is analogous to a (complex) quantum wave-function.
A final thermodynamical trajectory can be defined by finding the maximum
of this distribution. In the formal limit\footnote{This is formal because as noted earlier $\Delta t$ should be smaller
than $t_{r}$ for convergence of $IL$} of $\Delta t$ large the two methods coincide in exactly the same
sense that quantum mechanics becomes classical mechanics as $\hbar\rightarrow0$.
The dependence of results for the path integral approach on $\Delta t$
remains unexplored and is potentially very interesting as any deviation
will amount to a ``quantum'' effect on the thermodynamics. Since
it is unclear at this point whether this second approach yields significantly
superior results experimentally we shall take advantage of the formal
limit of large $\Delta t$ to make the mathematics more tractable.
In this ``weak noise'' limit many useful results are available as
we shall see in the next section (see also \cite{Kle14}).

It is interesting to note that the consistency distribution for thermodynamical
variables $\lambda$ introduced by the author resembles conceptually
the stochastic thermodynamics of Seifert and others (see, for example
\cite{seifert2012stochastic}) who assume that thermodynamical variables
have an associated density which satisfies a Fokker Planck equation.
Indeed the consistency distribution satisfies a Wick rotated Schrödinger
equation and hence also defines a continuous Markov process. A detailed
comparison of the two approaches would be interesting as the approach
outlined here provides a possible statistical physics underpinning
to the more practical approach taken in stochastic thermodynamics.
Indeed the information loss Lagrangian discussed above is precisely
defined by the underlying dynamical system and the choice of the slow
variables of that system.

\section{Thermodynamics, equilibration and forced dissipative generalizations}

A subject of great interest in practical prediction problems concerns
the physical factors controlling the fundamental predictability time
limit. As a direct consequence the asymptotic convergence of trial
densities toward a statistical equilibrium is of central interest.
In terms of the path integral formalism of the previous section this
involves the large time behavior of the maxima of the consistency
distribution. The consequent implicit deviations from the equilibrium
density in this situation also give direct information as to which
patterns of slow variables are most predictable.

As is noted in \cite{Kle14}, the mathematics of this situation are
considerably simplified when the formal limit of large $\Delta t$
is considered. In general one can then factorize the consistency distribution
as
\begin{equation}
\psi(\lambda,t)=\exp\left(-\Delta t\left(f_{s}(\lambda(t))-f_{s}(\lambda(0))\right)\right)\rho(\lambda,t)\label{prod}
\end{equation}
where $f_{s}$ satisfies the stationary Hamilton-Jacobi equation\footnote{The Hamiltonian involved is that naturally associated with the Lagrangian
$\mathscr{L}$}:
\begin{equation}
\left(\frac{\partial f_{s}}{\partial\lambda_{k}}(\lambda)+M_{k}(\lambda)\right)^{t}h^{-1}\left(\lambda\right)\left(\frac{\partial f_{s}}{\partial\lambda_{k}}(\lambda)+M_{k}(\lambda)\right)=\phi\left(\lambda\right)\label{HamJ}
\end{equation}

In this weak noise limit of large $\Delta t$, $\rho$ is approximately
the density for a time dependent multivariate Onstein-Uhlenbeck (OU)
process and the maximum $\alpha$ of the density $\rho$ with respect
to $\lambda$ satisfies \uline{for all times} the equation
\begin{equation}
\dot{\alpha}_{i}=h_{ik}^{-1}(\alpha)\left(\frac{\partial f_{s}}{\partial\lambda_{k}}(\alpha)+M_{k}(\alpha)\right)\label{Oet}
\end{equation}

This equation has been analyzed in depth in \cite{turkington2012optimization}
as his ``stationary closure'' thermodynamics. It is shown there
to be of the GENERIC form proposed by Öttinger \cite{oett} for irreversible
thermodynamics. Note though that the RHS is determined analytically
once slow variables are identified rather than being empirically determined.
Furthermore in the case discussed here as opposed to \cite{turkington2012optimization},
we need to also consider the first factor on the RHS of (\ref{prod})
when calculating the peak of the consistency distribution to obtain
the thermodynamics. Simple cases show that this causes an additional
realistic ``spin up'' effect but it also modifies the nature of
the equilibration process as will be discussed below.

In order that a thermodynamical trajectory converges asymptotically
to an equilibrium value it is easy to see that this dynamical equation
must also converge to some value $\alpha^{*}$. Of course it remains
to be shown that this convergence actually occurs for a particular
set of initial conditions. Nevertheless we have the general result
\begin{prop}
Assuming $h$ is invertible at $\alpha^{*}$ then equilibrium can
occur at $\alpha=\alpha^{*}$ if and only if $\phi(\alpha^{*})=0$.
\end{prop}

\begin{proof}
If $\alpha^{*}$ is an equilibrium point then we must have
\[
h^{-1}(\alpha^{*})\left(\frac{\partial f_{s}}{\partial\lambda_{k}}(\alpha^{*})+M_{k}(\alpha^{*})\right)=0
\]

however the function $f_{s}$ satisfies the stationary Hamilton-Jacobi
equation
\[
\left(\frac{\partial f_{s}}{\partial\lambda_{k}}(\alpha^{*})+M_{k}(\alpha^{*})\right)^{t}h^{-1}\left(\alpha^{*}\right)\left(\frac{\partial f_{s}}{\partial\lambda_{k}}(\alpha^{*})+M_{k}(\alpha^{*})\right)=\phi\left(\alpha^{*}\right)
\]

so the RHS must vanish. Conversely suppose $\phi\left(\alpha^{*}\right)=0$
then the LHS of the previous HJ equation must vanish. Since $h$ is
invertible at $\alpha^{*}$ and it is a covariance matrix of the slow
variables then $h^{-1}\left(\alpha^{*}\right)$ is positive definite
implying that the RHS of (\ref{Oet}) must vanish at this point.
\end{proof}
We now discuss separately the cases where the underlying dynamics
are Hamiltonian and where they are of a more general autonomous type.

\subsection{Hamiltonian dynamics}

In the non-singular case $\alpha^{*}$ is unique and zero.
\begin{prop}
Hamiltonian systems have a unique equilibrium $\alpha^{*}=0$ providing
the matrix
\[
Q_{ij}\equiv\left\langle (LA_{i})(LA_{j})\right\rangle 
\]

is everywhere non-singular. Furthermore we have $M(0)=\nabla f_{s}(0)=0$.
\end{prop}

\begin{proof}
It is obvious from it's definition that $\phi(0)=0$ and $\alpha^{*}=0$
is an equilibrium. Furthermore we have
\[
M_{i}=-\alpha_{j}\left\langle A_{i}(LA_{j})\right\rangle 
\]

implying $M(0)=0$. This then implies that $\nabla f_{s}(0)=0$ using
the last proposition. Now by definition we have
\[
\phi(\alpha^{*})=\alpha_{i}^{*}Q_{ij}\alpha_{j}^{*}
\]

and it is trivial to show that
\[
Q_{ij}=\left\langle \left(LA_{i}-M_{i}\right)\left(LA_{j}-M_{j}\right)\right\rangle +M_{i}M_{j}
\]

where the first matrix on the RHS is a covariance matrix $C_{ij}$.
Thus
\[
\alpha_{i}Q_{ij}\alpha_{j}=\alpha_{i}C_{ij}\alpha_{j}+\left(\alpha_{i}M_{i}\right)^{2}\geq0
\]

implying $Q$ is positive definite since it is invertible. Thus $\phi(\alpha^{*})=0$
iff $\alpha^{*}=0$.
\end{proof}
If we assume that convergence of $\alpha$ does occur then for large
$t$ the density $\rho$ is approximately that of a standard OU process.
In both cases the drift vector and noise covariance matrix are determined
analytically by $\Delta t$ and the moments $g$ and $M$ and the
function $\phi$ from equation (\ref{explicitlagr}). Given this situation
the asymptotic time behavior of the maximum of the consistency distribution
$\hat{\lambda}$ from can be determined from (\ref{prod}) using the
fact that the asymptotic $\rho$ is Gaussian with variance $\sigma$.
It is then easy to show that it satisfies the implicit equation
\begin{equation}
\hat{\lambda_{i}}=\alpha_{i}-\sigma\frac{\partial f_{s}}{\partial\lambda_{i}}(\hat{\lambda})\label{implicit}
\end{equation}

where $\sigma$ is the equilibrium covariance matrix of the standard
multivariate OU process. Now we have seen above that both $M$ and
$\nabla f_{s}$ vanish asymptotically and so may be expanded for large
$t$ as
\begin{align}
M & =J\alpha'\nonumber \\
\nabla f_{s} & =G\alpha'\label{diss}
\end{align}

where it may be shown that $J_{ij}=\frac{\partial M_{i}}{\partial\lambda_{j}}(0)=-J_{ji}$
and further that $G$ symmetric and positive definite. The former
matrix can be easily obtained analytically. The latter can be obtained
(as discussed in \cite{turkington2012optimization}) by considering
the small perturbation limit of the stationary HJ equation. This results
in the Riccati equation
\begin{align}
\left(G+J\right)^{t}h^{-1}(0)\left(G+J\right) & =N\label{Riccati}\\
N_{ij} & \equiv\frac{\partial^{2}\phi}{\partial\lambda_{i}\partial\lambda_{j}}(0)
\end{align}

This equation may be solved for $G$ using standard techniques. The
solution will be unique providing that the matrix
\[
A\equiv h^{-1}(0)\left(G+J\right)
\]

has only eigenvalues with negative real parts. This condition is known
as a stabilizing criteria in Riccatti parlance and in our case this
means that all perturbations near equilibrium eventually equilibrate.

We may now determine $\sigma$ by solving the appropriate Lyupanov
equation for the OU process (see \cite{gard} subsection 4.4.6)
\begin{align}
A\sigma+\sigma A^{t} & =h^{-1}(0)\label{Lyupanov}
\end{align}

Inserting (\ref{diss}) into (\ref{implicit}) we can solve the latter
to first order as 
\[
\hat{\lambda}'=\left(I+\sigma G\right)^{-1}\alpha'
\]

which will be a consistent perturbative solution providing $\sigma$
and $G$ are no larger than order $1$. This therefore determines
the asymptotic thermodynamical behavior of the system. It is easily
seen that this satisfies the relaxation equation
\begin{align}
\frac{\partial\hat{\lambda}'}{\partial t} & =P^{-1}h^{-1}(0)\left(G+J\right)P\hat{\lambda}'\label{decay}\\
P & \equiv\left(I+\sigma G\right)\nonumber 
\end{align}

The relaxation equation for $\alpha'$ is the same except that $P$
is set to unity. Thus the two relaxations have a set of decay modes
with the same eigenvalues but differing eigenvectors. The slowest
decaying eigenvector gives the most predictable mode for the system
while the corresponding eigenvalue gives the fundamental predictability
limit time scale. In order to better physically understand this mode
it will be often useful to take the Legendre transform and solve instead
for the relevant slow variable moment modes. This transform can be
considered to be a trial density manifold co-ordinate transformation
(see section 2 above). Given the tensor transformational properties
of $h$, $M$ and $\phi$ and the properties of the unique equilibrium
point, one may verify straightforwardly that equations (\ref{Riccati}),
(\ref{Lyupanov}) and (\ref{decay}) hold where one uses the transformed
tensors in their new co-ordinates and performs gradient calculations
also in these new co-ordinates and evaluates the resulting matrices
again in the new co-ordinates for the equilibrium point. Indeed the
entire framework is conveniently and attractively covariant. This
fact will allow us to use more convenient moment co-ordinates in the
next section.

\subsection{Autonomous dynamics}

The information loss formalism of section 2 has been generalized to
the autonomous case in \cite{Kle15}. The third equation of (\ref{IL})
is modified to
\begin{align}
R & =\left(\partial_{t}+L\right)\log\hat{\varrho}+\frac{\partial C_{k}}{\partial x_{k}}=\dot{\lambda}_{i}\left(A_{i}-a_{i}\right)+\lambda_{i}LA_{i}-\beta LF+\frac{\partial C_{k}}{\partial x_{k}}\label{aut}\\
\hat{\varrho}_{eq} & =D\exp\left(-\beta F(x)\right)\nonumber 
\end{align}

In the Hamiltonian case the divergence of $C$ vanishes as does $LF$
since the Gibbs density exponent has a zero Poisson bracket with the
Hamiltonian. Neither of these conditions are satisfied for a general
autonomous system. Nevertheless if the dynamical system possesses
an equilibrium density (which we assume) then it satisfies the steady
state Liouville equation:
\[
\left(L+\frac{\partial C_{k}}{\partial x_{k}}\right)\hat{\varrho}_{eq}=0
\]

which implies if we further assume that $\hat{\varrho}_{eq}$ is nowhere
vanishing that the sum of the final two terms in (\ref{aut}) vanish.
The generic form of $R$ is thus the same as the Hamiltonian case
and the conclusions of of the previous subsection carry over in their
entirety. We thus have a consistent formalism in the sense that $\lambda=0$
describes the only possible thermodynamical equilibrium and it corresponds
with a trial density of $\hat{\varrho}_{eq}$.

Of course in a certain sense we have buried the issue since knowledge
of $h$, $M$ and $\phi$ requires knowledge of $F$. In the Hamiltonian
case this is generally quite accurately known\footnote{It is interesting to note though that strictly speaking we only know
that the Gibbs density is an invariant and there may be a large number
of these within a Hamiltonian system. Thus the form of the equilibrium
density must eventually be deduced either empirically or by appeal
to statistical independence and locality arguments (see \cite{ll80}
section 2).} but in the more general autonomous case it requires observational
estimation. Put another way, if we assume that $F(x)$ has a particular
restricted form (for example quadratic for a Gaussian trial density)
then the vanishing of the final two terms in (\ref{aut}) only occurs
approximately and the better the match of $F$ to observation then
the better this approximation becomes.

\subsection{General solutions using a numerical method}

It would of course be desirable to have at least a numerical method
for exploring thermodynamical trajectories far from equilibrium. This
can be achieved in principle as follows: In the limit of $\Delta t$
large the relevant action $\mathcal{S}_{cl}$ for determining thermodynamical
behavior is that from the extremal trajectory. We evaluate each possible
endpoint according to this value and choose that with the minimal
value as our thermodynamical point. Evidently we therefore seek an
endpoint satisfying
\[
\frac{\partial\mathcal{S}_{cl}}{\partial\lambda_{i}}=0
\]
 Now as noted in \cite{feynman1965quantum} p28, the LHS here is simply
the conjugate momentum $p_{i}\equiv\frac{\partial\mathscr{L}}{\partial\dot{\lambda}_{i}}$
at this endpoint. Thus we consider an extremal trajectory with final
conjugate momentum zero and a specified initial condition $\lambda(0)$
and extract $\lambda(T)=\hat{\lambda}(T)$. The trajectories involved
are, of course, solutions of the Euler-Lagrange equations for the
problem and as noted in \cite{Kle14} have the form of ``charged
particle'' geodesics in a manifold with Riemannian metric $h$ with
the particle subject to an external ``electromagnetic'' vector $(M,\phi)$.
Given the boundary condition requirements imposed, the equations are
most conveniently expressed using coupled equations from standard
Hamiltonian mechanics: 
\begin{align*}
\dot{\lambda}_{i} & =\frac{\partial\mathcal{H}}{\partial p_{i}}=h_{ij}^{-1}(p_{j}+M_{j})\\
\dot{p}_{i} & =-\frac{\partial\mathcal{H}}{\partial\lambda_{i}}=\frac{\partial\phi}{\partial\lambda_{i}}+\frac{\partial M_{k}}{\partial\lambda_{i}}h_{kj}^{-1}(p_{j}+M_{j})+\frac{1}{2}(p_{j}+M_{j})\frac{\partial h_{jl}^{-1}}{\partial\lambda_{i}}(p_{l}+M_{l})\\
\lambda(0) & =\lambda_{0}\\
p(T) & =0
\end{align*}

where the Hamiltonian is obtained in the usual manner for the Lagrangian
$\mathscr{L}$ discussed in section 2. This is a linear ordinary differential
equation boundary value problem albeit with an unusual endpoint condition
in $p$ rather than $\lambda$. There is a large literature on a wide
variety of different numerical methods for the solution of such problems
with a classical reference being \cite{ascher1994numerical}. It is
worth noting that in the usual classical mechanics one typically specifies
$\lambda$ and $p$ at the initial time but of course here the initial
$p$ is unknown as opposed to a fixed final $p$ of zero.

In a follow up publication to the present one the practicalities of
various numerical methods of solution will be explored in a range
of different turbulence problems. An attractive feature of the approach
outlined is that avoids the need to solve a multi-dimensional non-linear
Hamilton-Jacobi partial differential equation which is required to
integrate equation (\ref{Oet}) directly.

\section{Gaussian trial densities and fluid dynamical systems}

In order to allow for a practical implementation of the formalism
outlined in the previous section we require certain moments of the
trial density as well as the partition function $Z(\lambda,\beta)$.
Unless this density is chosen from a manageable family then analytical
expressions for these quantities may be difficult to obtain. In many
problems of interest in fluid dynamics, numerical simulations give
strong evidence for densities with only a small amount of non-Gaussian
behavior (see, for example, the following work by the author and collaborators:
\cite{kleeman05a}, \cite{kleeman2012nonequilibrium} and \cite{klee04d}).
Motivated by these considerations we consider the case that the trial
density is a multivariate Gaussian. Furthermore we shall restrict
initially this density in such a way that the covariance of slow and
fast variables is always zero. This is for pedagogical clarity here
but is often empirically justified. It could be relaxed if required.

The general trial density may therefore be written as
\begin{eqnarray*}
\hat{\varrho} & = & \left\{ 2\pi\right\} ^{-\frac{n}{2}}\sqrt{g}\exp\left[-\frac{1}{2}g^{ij}\left(x_{i}-\mu_{i}\right)\left(x_{j}-\mu_{j}\right)\right]\\
g & \equiv & \det\left(g_{ij}\right)=\left(\det\left(g^{ij}\right)\right)^{-1}\\
g_{ij} & = & \left\langle \left(x_{i}-\mu_{i}\right)\left(x_{j}-\mu_{j}\right)\right\rangle _{\hat{\varrho}}\\
\mu_{i} & = & \left\langle x_{i}\right\rangle _{\hat{\varrho}}
\end{eqnarray*}

where the indices range over the full set of $n$ variables with slow
variables occupying the first $m$ slots. Indices are raised, lowered
and contracted in the usual tensor manner. The covariance matrix $\left\{ g_{ij}\right\} $
is block diagonal with respect to fast and slow variables as therefore
is the inverse $\left\{ g^{ij}\right\} $. A convenient co-ordinatization
of the trial density manifold here is provided by the slow variable
covariance matrix and mean vector i.e. all moments of order 1 and
2. There are clearly $\frac{1}{2}m(m+3)$ in total. As was noted at
the end of section 3 these co-ordinates are the physically transparent
ones for asymptotic predictability analysis and are also more mathematically
tractable. The co-ordinate transformation to the $\lambda$ used in
Section 2 and 3 is a non-linear Legendre transformation which can
be obtained explicitly using Cramer's rule for $\left\{ g^{ij}\right\} $.
The $A(x)$ of section 2 equation (\ref{MAXENT}) are obviously first
and second order polynomials of the $x_{i}$. Given the covariant
form of the formalism it is not necessary to know the nature of this
transformation explicitly.

In what follows indices will be assumed to range over $n$ values
however time derivatives will only be non-zero for slow variable co-ordinates.
This device allows a clean interpretation in terms of the Lagrangian
of section 2 albeit with transformed co-ordinates.

One of the motivations of the current theoretical development is the
study of fluid systems and their predictability. Consequently we shall
consider a restricted class of autonomous dynamical systems which
however cover a large variety of practical fluid systems. In particular
we shall assume that the base dynamical system satisfies
\begin{equation}
\frac{dx_{i}}{dt}=C_{i}(x)=B_{i}^{jk}x_{j}x_{k}+H_{i}^{k}x_{k}+X_{i}\label{dynamicalsystem}
\end{equation}

and that when $H=0$ and $X=0$ we have Hamiltonian dynamics. Thus
the Liouvillean condition applies to the quadratic term on the RHS
and so
\begin{equation}
B_{i}^{ij}=-B_{i}^{ji}\label{Liouville}
\end{equation}

with the summation convention assumed here and later\footnote{Without loss of generality one can assume that $B$ is symmetric in
the upper two indices and therefore also that
\[
B_{i}^{ij}=0
\]
}. Due to conservation principles in the underlying dynamical system
$B$ may satisfy other identities beyond (\ref{Liouville}). The second
term on the RHS of (\ref{dynamicalsystem}) can represent both linear
dissipation as well as linear instability due to an imposed mean flow.
The final term represents external forcing. In order to carry out
the program outlined in the previous sections we require the Lagrangian
$\mathscr{L}=\frac{1}{2}\left\langle R^{2}\right\rangle $ where the
expectation is calculated with respect to the assumed Gaussian trial
density. This latter assumption enables this to be done in a relatively
straightforward (albeit tedious) fashion. Details for the Hamiltonian
and more general autonomous case are in Appendix A. The result for
the Hamiltonian case with $H=X=0$ is 
\begin{eqnarray}
2\mathscr{L}=\left\langle R^{2}\right\rangle  & = & \left(\dot{\mu}_{i}-F_{i}\right)g^{ij}\left(\dot{\mu}_{j}-F_{j}\right)+\frac{1}{4}\left(\dot{g}_{ij}-2Q_{ij}\right)\left(g^{ik}g^{jl}+g^{il}g^{jk}\right)\left(\dot{g}_{kl}-2Q_{kl}\right)\nonumber \\
 &  & -2S^{i}\left(\dot{\mu}_{i}-F_{i}\right)+\Phi\label{Gaussianlagrangian}\\
F_{i} & \equiv & B_{i}^{kl}\mu_{k}\mu_{l}\nonumber \\
Q_{ij} & \equiv & \mu_{l}g_{kj}\left(B_{i}^{kl}+B_{i}^{lk}\right)\nonumber \\
S^{i} & \equiv & g^{ij}g_{ln}B_{j}^{ln}\nonumber \\
\Phi & \equiv & B_{i}^{jk}B_{l}^{mn}g^{il}\left(g_{kj}g_{nm}+g_{km}g_{nj}+g_{jm}g_{nk}\right)\nonumber \\
 &  & +\left(B_{i}^{jk}+B_{i}^{kj}\right)\left(B_{j}^{il}+B_{j}^{li}\right)g_{lk}\nonumber 
\end{eqnarray}

It should be stressed that the implied summations in this result are
over $n$ (the total number of variables in the original dynamical
system) however only the slow variables have non-zero time derivatives.
The result for the more general autonomous case has generically the
same form as (\ref{Gaussianlagrangian}) but with modified versions
of the fields $F$, $Q$, $S$ and $\Phi$.

Note also that the result has the quadratic form in time derivatives
derived in section 2 but with respect to a (Legendre) transformed
set of co-ordinates namely the slow variable means and covariances.
We can immediately read off the required $h$,$M$ and $\phi$ in
the transformed co-ordinates. Of course to use this in section 3 for
an asymptotic analysis, the values of these quantities and their various
partial derivatives must be evaluated at the point of equilibrium
$\lambda=0$ which needs to be expressed in terms of the new moment
co-ordinates. The covariant nature of our formalism also allows a
straightforward use of the results of subsection 3.3.

\subsection{Further simplifications}

In the case of many turbulent fluid systems, empirical evidence suggests
that to a good approximation the spectral modes are not strongly correlated
(see, for example, \cite{kleeman2012nonequilibrium} and \cite{TCT15})
and furthermore the variances of real and imaginary parts of these
modes are approximately equal. This allows for considerable simplification
of (\ref{Gaussianlagrangian}). We illustrate this for the Hamiltonian
case but some of the results carry over to the more realistic forced
dissipative turbulence case as will be discussed in the companion
paper. Denoting the variances by $b_{i}$ it is easily shown that
(\ref{Gaussianlagrangian}) becomes
\[
2\mathscr{L}=IL_{1}+IL_{2}+IL_{3}+IL_{4}+IL_{5}
\]

with
\begin{align*}
IL_{i} & =\sum_{i}\left(\dot{\mu}_{i}-F_{i}\right)b_{i}^{-1}\left(\dot{\mu}_{i}-F_{i}\right)\\
IL_{2} & =\sum_{i,j,m,n}\left[B_{i}^{jn}B_{i}^{jm}b_{j}b_{i}^{-1}+B_{i}^{jn}B_{j}^{im}\right]\mu_{n}\mu_{m}\\
IL_{3} & =\frac{1}{2}\sum_{i}\left\{ \left(\frac{\dot{b}_{i}}{b_{i}}\right)^{2}-4\sum_{m}\left(\frac{\dot{b}_{i}}{b_{i}}\right)B_{i}^{im}\mu_{m}\right\} \\
IL_{4} & =\sum_{i,j,k}\left\{ b_{i}^{-1}\left[\left(B_{i}^{kk}b_{k}\right)^{2}+2\left(B_{i}^{jk}\right)^{2}b_{j}b_{k}\right]+4B_{i}^{jk}B_{j}^{ik}b_{k}\right\} \\
IL_{5} & =\sum_{i,l}b_{i}^{-1}b_{l}B_{i}^{ll}\left(\dot{\mu}_{i}-F_{i}\right)
\end{align*}

where we have now explicitly included summations. As noted, many fluid
systems can be written in spectral form and then the quadratic non-linear
terms reduce to the form of equation (\ref{dynamicalsystem}) but
with an interesting restriction on the indices of $B_{i}^{jk}$. This
is often termed (see \cite{salmon-book} Chapter 5) a selection rule
and takes the form
\[
i=j+k
\]

where the indices may regarded profitably as integer vectors. This
rule follows simply from the Fourier functions orthogonality condition.
The three modes satisfying this selection rule are the modal triads
of standard turbulence theory and energy can transfer from any two
of them to a third. In many systems of interest it will also be the
case that the zero wave-number component will be an invariant of the
system and so can be set to zero without loss of generality. These
two rules have the effect of eliminating the second term in $IL_{3}$
and the final term in $IL_{4}$.

In addition the fact that these systems can be formulated in terms
of complex variables means that 
\[
B_{l}^{k_{R}k_{R}}=-B_{l}^{k_{I}k_{I}}
\]

where the suffices $R$ and $I$ refer to real and imaginary variables.
The assumption that $b_{k_{R}}=b_{k_{I}}$ then implies the first
term of $IL_{4}$ as well as all of $IL_{5}$ are both zero.

Since the dynamical system can be specified by a complex equation
rather than the real one in (\ref{dynamicalsystem}), the real and
imaginary parts of the Fourier coefficients can be conveniently assigned
a $Z_{2}$ value which may be appended to the spectral vectors and
the selection rule still holds then for the real $B_{i}^{jk}$. A
similar device may be used when dealing with vertical modes . Consider
now the term $IL_{2}$: The selection rule implies that we can write
this as
\begin{equation}
IL_{2}=\sum_{i=j+n}\left\{ b_{j}b_{i}^{-1}\left(B_{i}^{jn}\right)^{2}\left(\mu_{n}\right)^{2}+B_{i}^{jn}B_{j}^{i-n}\mu_{n}\mu_{-n}\right\} \label{IL2}
\end{equation}

where the summation extends over all triads satisfying $i=j+n$. The
reality condition for the complex dynamical variables also implies
that 
\[
\mu_{-n}=\pm\mu_{n}\equiv p(n)\mu_{n}
\]

where the ``parity'' $p(n)$ is $+1$ when $n$ refers to a real
part of a mode and $-1$ for the imaginary part. Summarizing we can
write
\begin{align*}
2\mathscr{L} & =\sum_{i}\left(\dot{\mu}_{i}-F_{i}\right)b_{i}^{-1}\left(\dot{\mu}_{i}-F_{i}\right)+\frac{1}{2}\sum_{i}\left(\frac{\dot{b}_{i}}{b_{i}}\right)^{2}\\
 & +\sum_{i=j+n}\left\{ b_{j}b_{i}^{-1}\left(B_{i}^{jn}\right)^{2}+p(n)B_{i}^{jn}B_{j}^{i-n}\right\} \left(\mu_{n}\right)^{2}+2\left(B_{i}^{jn}\right)^{2}b_{i}^{-1}b_{j}b_{n}
\end{align*}

This considerable simplification allows for a tractable analysis of
a wide variety of Hamiltonian fluid systems.

\section{Summary and future work}

In recent years a number of different approaches to the problem of
statistical disequilibrium have been proposed. These have ranged from
empirically oriented thermodynamic approaches to those based on the
evolution of the probability density of slow variables within a statistical
system. The author proposed an approach recently of the second type
which was derived from a path integral approach using a generalized
Boltzmann principle. The resulting formalism is attractive theoretically
as it is directly related to conventional quantum mechanics once a
Wick rotation is performed.

The practical application of the method was however left mainly unexplored.
In particular a method to calculate thermodynamical trajectories was
only briefly mentioned. Furthermore the formalism only applied to
idealized Hamiltonian dynamics rather than the forced dissipative
systems underlying realistic turbulence. The current presentation
attempts to remedy these practical deficiencies in several ways:
\begin{enumerate}
\item A general method is described that allows a description of near equilibrium
behavior. The linear method involved relies on the solution of Riccatti
and Lyupanov equations. The matrices derived allow for a general study
of predictability limits within turbulent systems.
\item The formalism is extended in a natural manner to consider dynamical
systems described by autonomous equations. The more general formalism
has a unique thermodynamic equilibrium.
\item The entire framework is covariant with respect to the the choice of
thermodynamical variables or co-ordinates. This enables a choice of
convenient co-ordinates such as moments or inverse temperatures.
\item A numerical method is described for exploring the thermodynamical
trajectories far from equilibrium. Unlike the near equilibrium situation
different initial conditions must be solved for separately in this
method.
\end{enumerate}
The thermodynamics derived is similar to that proposed in the GENERIC
framework by Öttinger and co-workers but in contrast is an ab initio
derivation from the underlying dynamical system \uline{providing}
that an appropriate family of approximating slow variable probability
densities are identified. The method is also conceptually similar
to the stochastic thermodynamics of Seifert and co-workers in that
thermodynamical variables are associated with a non-negative ``consistency''
distribution which evolves according to a stochastic equation\footnote{The equation here is a Wick rotated Schrödinger equation which in
a certain limit is controlled by a Fokker Planck equation. Stochastic
thermodynamics has a thermodynamical density evolving according to
a Fokker Planck equation.}.

In order to complete the program outlined in the four points above,
an appropriate approximating density family of ``trial'' densities
must be identified and it needs to be a good approximation of observed
slow variable densities for appropriate thermodynamic parameters.
Furthermore the analytical calculations involved rely on the trial
densities having tractable behavior with respect to the evaluation
of partition functions and moments. In section 4 we showed how this
can be done with the choice of Gaussian trial densities. Different
trial densities could be conceived of for different dynamical systems.
Section 4 should provide a template for this however a viewing of
the algebra involved there suggests that there may be formidable calculations
involved.

While the present paper describes general theoretical machinery it
is clearly important to validate this in interesting statistical/turbulent
systems using numerical methods. These will be performed in a follow
up publication. It would also be very useful to directly compare results
with other more empirical thermodynamical approaches since these are
after all more closely connected to concrete applications.\pagebreak{}

\appendix

\section{The Lagrangian for Gaussian trial densities}

We consider first the Hamiltonian case $H=X=0$ and then generalize
below: The Liouville residual for the trial densities is
\begin{align*}
R & =\left(\frac{\partial}{\partial t}+A_{i}\frac{\partial}{\partial x_{i}}\right)\hat{l}\\
\hat{l} & \equiv\log\hat{\varrho}
\end{align*}

Using the chain rule we have 
\[
\frac{\partial}{\partial t}=\frac{\partial g_{ij}}{\partial t}\frac{\partial}{\partial g_{ij}}+\frac{\partial\mu_{i}}{\partial t}\frac{\partial}{\partial\mu_{i}}
\]

so
\begin{eqnarray}
\frac{\partial\hat{l}}{\partial t} & = & \dot{g}_{ij}\left(\frac{\partial\ln\sqrt{g}}{\partial g_{ij}}-\frac{1}{2}\frac{\partial g^{kl}}{\partial g_{ij}}\left(x_{k}-\mu_{k}\right)\left(x_{l}-\mu_{l}\right)+\dot{\mu}_{i}g{}^{ij}\left(x_{j}-\mu_{j}\right)\right)\label{Hamcalc}
\end{eqnarray}

where we have used the fact that $r$ and $s$ are fixed.

Now
\begin{equation}
\frac{\partial\ln\sqrt{g}}{\partial g_{ij}}=\frac{1}{2}g^{-1}\frac{\partial g}{\partial g_{ij}}=-\frac{1}{2}g^{ij}\label{detderiv}
\end{equation}

by standard results for the derivative of the determinant. Similarly
\begin{equation}
\frac{\partial g^{kl}}{\partial g_{ij}}=-g^{ik}g^{jl}\label{invertderiv}
\end{equation}

again by a standard matrix result. Introduce now the following score
variables
\begin{eqnarray}
\frac{\partial\hat{l}}{\partial\mu_{i}} & = & g^{ij}\left(x_{j}-\mu_{j}\right)\equiv U^{i}\nonumber \\
\frac{\partial\hat{l}}{\partial g_{mp}} & = & \frac{1}{2}\left(g^{im}g^{jp}\left(x_{i}-\mu_{i}\right)\left(x_{j}-\mu_{j}\right)-g^{mp}\right)\equiv V^{mp}=\frac{1}{2}\left(U^{m}U^{p}-g^{mp}\right)\label{relationship}
\end{eqnarray}

where we used (\ref{detderiv}) and (\ref{invertderiv}). These variables
make the calculation of expectation values considerably easier: It
is easily verified that 
\[
\left\langle U^{i}\right\rangle =\left\langle V^{ij}\right\rangle =0
\]

as is the usual score variable situation. It is also obvious from
above that
\begin{equation}
\partial_{t}\hat{l}=\dot{g}_{mp}V^{mp}+\dot{\mu}_{m}U^{m}\label{timederiv}
\end{equation}

Because the trial density is Gaussian, the score variables are easily
shown to satisfy
\[
\left\langle U^{i}U^{j}\right\rangle =g^{ij}
\]

and
\[
\left\langle V^{ij}U^{k}\right\rangle =0
\]

In addition we have
\begin{eqnarray*}
\left\langle V^{ij}V^{kl}\right\rangle  & = & \frac{1}{4}\left\langle \left(U^{i}U^{j}-g^{ij}\right)\left(U^{k}U^{l}-g^{kl}\right)\right\rangle \\
 & = & \frac{1}{4}\left[\left\langle U^{i}U^{j}U^{k}U^{l}\right\rangle -g^{ij}g^{kl}\right]\\
 & = & \frac{1}{4}\left[g^{ik}g^{jl}+g^{il}g^{jk}\right]
\end{eqnarray*}

where we are using the fourth order Gaussian moment formulae for variables
of zero mean. Other higher Gaussian moment identities are
\begin{eqnarray*}
\left\langle U^{i}U^{j}U^{k}U^{l}\right\rangle  & = & g^{ik}g^{jl}+g^{il}g^{jk}+g^{ij}g^{kl}\\
\left\langle U^{i}U^{j}U^{k}U^{l}U^{m}U^{n}\right\rangle  & = & \sum_{Binary\;Partitions}g^{ij}g^{kl}g^{mn}
\end{eqnarray*}

where binary partitions means every partitioning of the set $ijklmn$
into three two member groups (15 in all). Finally the expectations
of the product of an odd number of $U$ is zero. The underlying dynamical
system time tendency may be re-expressed in terms of score variables
as
\begin{eqnarray}
C_{i}=B_{i}^{kl}x_{k}x_{l} & = & B_{i}^{kl}\left(\mu_{k}+g_{kj}U^{j}\right)\left(\mu_{l}+g_{lp}U^{p}\right)\nonumber \\
 & = & B_{i}^{kl}\mu_{k}\mu_{l}+\mu_{l}g_{kj}U^{j}\left(B_{i}^{kl}+B_{i}^{lk}\right)+B_{i}^{kl}g_{kj}g_{lp}U^{j}U^{p}\nonumber \\
 & \equiv & F_{i}+Q_{ij}U^{j}+P_{ijp}U^{j}U^{p}\label{AA-1}
\end{eqnarray}

Note that $F_{i}$ is the dynamical system tendency term for the first
moments. We have now 
\[
A_{i}\frac{\partial}{\partial x_{i}}\hat{l}=-g^{il}(x_{l}-\mu_{l})A_{i}=-U^{i}A_{i}
\]

and so
\[
R=\dot{g}_{ij}V^{ij}+\left(\dot{\mu}_{i}-F_{i}\right)U^{i}-Q_{ij}U^{j}U^{i}-P_{ijp}U^{j}U^{p}U^{i}
\]

Using the fact that $Q_{ij}g^{ij}=0$ and (\ref{relationship}) this
can be written more compactly as 
\[
R=\left(\dot{g}_{ij}-2Q_{ij}\right)V^{ij}+\left(\dot{\mu}_{i}-F_{i}\right)U^{i}-P_{ijp}U^{j}U^{p}U^{i}
\]

The information loss can now be computed straightforwardly using product
of score variable expectations calculated above and the result is
given in the text as equation (\ref{Gaussianlagrangian}). The identities
for $S$ and $\Phi$ in that equation use (\ref{Liouville}) repeatedly
in their derivation.

Consider now the more general case. We need to consider now a more
general equation for $R$ as is evident from equation (\ref{aut}).
Given the assumption made earlier that allows us to drop the final
two terms in that equation it becomes convenient to write
\begin{align*}
\hat{l} & =\ln\sqrt{g}+\hat{l}_{s}+\hat{l}_{eq}+Const\\
\hat{l}_{eq} & \equiv-\frac{1}{2}r^{ij}\left(x_{i}-s_{i}\right)\left(x_{j}-s_{j}\right)
\end{align*}

where $r$ and $s$ are the equilibrium inverse covariances and means
respectively which are fixed and determined empirically. Now the autonomous
Liouville residual under the approximation assumption of subsection
3.2 is given by
\[
R=\left(\frac{\partial}{\partial t}+C_{i}\frac{\partial}{\partial x_{i}}\right)\hat{l}_{s}+\left(\ln\sqrt{g}\right)_{t}
\]
The LHS of equation (\ref{Hamcalc}) becomes $\frac{\partial\hat{l}_{s}}{\partial t}+\left(\ln\sqrt{g}\right)_{t}$
and the RHS is unchanged.

The underlying dynamical system time tendency generalizes to
\begin{eqnarray}
C_{i} & = & F_{i}+Q_{ij}U^{j}+P_{ijp}U^{j}U^{p}\label{AA}\\
F_{i} & \equiv & B_{i}^{kl}\mu_{k}\mu_{l}+H_{i}^{k}\mu_{k}+X_{i}\\
Q_{ij} & \equiv & \mu_{l}g_{kj}\left(B_{i}^{kl}+B_{i}^{lk}\right)+H_{i}^{k}g_{kj}\\
P_{ijp} & \equiv & B_{i}^{kl}g_{kj}g_{lp}
\end{eqnarray}

Note now that $F_{i}$ is the dynamical system tendency term for mean
deviations from the equilibrium values.

The remaining term in $R$ may be now calculated in terms of score
variables:
\begin{align*}
C_{i}\frac{\partial}{\partial x_{i}}\hat{l}_{s} & =-\left(U^{k}Z_{k}^{i}+Y^{i}\right)C_{i}\\
Z_{k}^{i} & \equiv\left(\delta_{k}^{i}-r^{il}g_{lk}\right)\\
Y^{i} & \equiv r^{il}\left(\mu_{l}-s_{l}\right)
\end{align*}

and so
\begin{align*}
R & =\dot{g}_{ij}V^{ij}+\left(\dot{\mu}_{i}-N_{i}\right)U^{i}-M_{ij}U^{j}U^{i}-T_{ijp}U^{j}U^{p}U^{i}-\Omega\\
\Omega & \equiv Y^{k}F_{k}\\
N_{i} & \equiv Z_{i}^{k}F_{k}+Y^{k}Q_{ki}\\
M_{ij} & \equiv Z_{i}^{k}Q_{kj}+Y^{k}P_{kji}\\
T_{ijp} & \equiv Z_{i}^{n}P_{njp}
\end{align*}

Using (\ref{relationship}) this can be written more compactly as
\begin{align*}
R & =\left(\dot{g}_{ij}-2M_{ij}\right)V^{ij}+\left(\dot{\mu}_{i}-N_{i}\right)U^{i}-T_{ijp}U^{j}U^{p}U^{i}-\Psi\\
\Psi & \equiv\Omega+g^{ij}M_{ij}
\end{align*}

The Lagrangian can now be computed using the expectation values for
the score variables:
\begin{eqnarray*}
2\mathscr{L}=\left\langle R^{2}\right\rangle  & = & \left(\dot{\mu}_{i}-N_{i}\right)g^{ij}\left(\dot{\mu}_{j}-N_{j}\right)+\frac{1}{4}\left(\dot{g}_{ij}-2M_{ij}\right)\left(g^{ik}g^{jl}+g^{il}g^{jk}\right)\left(\dot{g}_{kl}-2M_{kl}\right)\\
 &  & -2S^{i}\left(\dot{\mu}_{i}-F_{i}\right)+\Phi\\
S^{i} & \equiv & T_{jkl}\left(g^{ik}g^{jl}+g^{il}g^{jk}+g^{ij}g^{kl}\right)=B_{j}^{rs}g^{ij}g_{sr}-2r^{nl}g_{lj}B_{n}^{ij}-r^{ni}g_{sr}B_{n}^{rs}\\
\Phi & \equiv & \Psi^{2}+T_{ijk}T_{lmn}\sum_{Binary\;Partitions}g^{ij}g^{kl}g^{mn}
\end{eqnarray*}

\bibliographystyle{plain}
\bibliography{/home/richard/Dropbox/refmerge}

\begin{thebibliography}{10}

\bibitem{ama00}
S.~Amari and H.~Nagaoka.
\newblock {\em Methods of Information Geometry}.
\newblock Translations of Mathematical Monographs, AMS, Oxford University
  Press, 2000.

\bibitem{ascher1994numerical}
U.~M. Ascher, R.~M.~M. Mattheij, and R.~D. Russell.
\newblock {\em Numerical solution of boundary value problems for ordinary
  differential equations}.
\newblock SIAM, 1995.
\newblock 593pp.

\bibitem{de1984irreversible}
S.~R. De~Groot and P~Mazur.
\newblock {\em Irreversible Thermodynamics}.
\newblock Dover, New York, 1984.

\bibitem{feynman1965quantum}
R.~P. Feynman, A.~R. Hibbs, and D.~F. Styer.
\newblock {\em Quantum mechanics and path integrals}.
\newblock McGraw-Hill New York, 1965.

\bibitem{frederiksen2004regularized}
J.~S Frederiksen and A.~G. Davies.
\newblock The regularized dia closure for two-dimensional turbulence.
\newblock {\em Geophysical and Astrophysical Fluid Dynamics}, 98(3):203--223,
  2004.

\bibitem{gard}
C.~W. Gardiner.
\newblock {\em Handbook of Stochastic Methods for Physics, Chemistry and the
  Natural Sciences}.
\newblock Springer, 2004.

\bibitem{jou2010extended}
D.~Jou, G.~Lebon, and J.~Casas-V\'azquez.
\newblock {\em Extended Irreversible Thermodynamics}.
\newblock Springer, 2010.

\bibitem{kleeman05a}
R.~Kleeman.
\newblock Limits, variability and general behaviour of statistical
  predictability of the mid-latitude atmosphere.
\newblock {\em J. Atmos. Sci.}, 65:263--275, 2008.

\bibitem{Kle11}
R.~Kleeman.
\newblock Information theory and dynamical system predictability.
\newblock {\em Entropy}, 13:612--649, 2011.

\bibitem{Kle14}
R.~Kleeman.
\newblock A path integral formalism for non-equilibrium {H}amiltonian
  statistical systems.
\newblock {\em J. Stat. Phys.}, 158(6):1271--1297, 2015.
\newblock DOI 10.1007/s10955-014-1149-x, arXiv:1307.1102.

\bibitem{Kle15}
R.~Kleeman.
\newblock A path integral formalism for the closure of autonomous statistical
  systems.
\newblock 2015.
\newblock arXiv:1503.04325.

\bibitem{klee04d}
R.~Kleeman and A.~J. Majda.
\newblock Predictability in a model of geostrophic turbulence.
\newblock {\em J. Atmos. Sci.}, 62:2864--2879, 2005.

\bibitem{kleeman2012nonequilibrium}
R.~Kleeman and B.~E. Turkington.
\newblock A nonequilibrium statistical model of spectrally truncated
  {B}urgers-{H}opf dynamics.
\newblock {\em Comm. Pure Appl. Math.}, 67(12):1905--1946, 2014.

\bibitem{ll80}
L.~D. Landau and E.~M. Lifshitz.
\newblock {\em Statistical physics}.
\newblock Pergamon Press, New York, 1980.

\bibitem{liboff2003kinetic}
R.~L. Liboff.
\newblock {\em Kinetic theory: classical, quantum, and relativistic
  descriptions}.
\newblock Springer, 2003.

\bibitem{luzzi2013predictive}
R.~Luzzi, {\'A}.~R. Vasconcellos, and J.~G. Ramos.
\newblock {\em Predictive statistical mechanics: a nonequilibrium ensemble
  formalism}, volume 122.
\newblock Springer, 2013.

\bibitem{mtv2}
A..~J. Majda, I.~Timofeyev, and E.~Vanden-Eijnden.
\newblock A mathematics framework for stochastic climate models.
\newblock {\em Comm. Pure Appl. Math.}, 54:891--974, 2001.

\bibitem{oett}
H.~C. Oettinger.
\newblock {\em Beyond equilibrium thermodynamics}.
\newblock Wiley-Interscience, Hoboken, New Jersey, 2005.

\bibitem{salmon-book}
R.~Salmon.
\newblock {\em Lectures on Geophysical Fluid Dynamics}.
\newblock Oxford Univ. Press, New York, 1998.

\bibitem{seifert2012stochastic}
U.~Seifert.
\newblock Stochastic thermodynamics, fluctuation theorems and molecular
  machines.
\newblock {\em Rep. Prog. Phys.}, 75(12):126001, 2012.

\bibitem{turkington2012optimization}
B.~Turkington.
\newblock An optimization principle for deriving nonequilibrium statistical
  models of {H}amiltonian dynamics.
\newblock {\em J. Stat. Phys}, 152:569--597, 2013.

\bibitem{TCT15}
B.~Turkington, Q-Y. Chen, and S.~Thalabard.
\newblock Coarse-graining two-dimensional turbulence via dynamical
  optimization.
\newblock {\em Nonlinearity}, 29(10):2961, 2016.

\bibitem{zubarev1996statistical}
D.~Zubarev, V.~Morozov, and G.~Ropke.
\newblock {\em {Statistical mechanics of nonequilibrium processes. Vol. 1:
  Basic concepts, kinetic theory}}.
\newblock Berlin: Akademie Verlag, 1996.

\bibitem{zubarev1997statistical}
D.~N. Zubarev, V.~Morozov, and G.~Ropke.
\newblock {\em Statistical mechanics of nonequilibrium processes. {R}elaxation
  and Hydrodynamic processes}, volume~2.
\newblock Berlin: Akademie Verlag, 1997.

\bibitem{zwanzig2001nonequilibrium}
R.~Zwanzig.
\newblock {\em {Nonequilibrium statistical mechanics}}.
\newblock Oxford University Press, USA, 2001.

\end{thebibliography}

\end{document}